\documentclass[11pt]{llncs}
\usepackage{amssymb,amsmath,enumerate,graphicx,hyperref}

\usepackage{tikz}
\usetikzlibrary{snakes,shapes}
\tikzstyle{tre}=[circle,draw,minimum size=3mm]
\tikzstyle{btre}=[circle,draw,minimum size=4.5mm]
\newcommand{\etq}[1]{%
\draw (#1) node {\tiny $#1$};}

\renewcommand{\leq}{\leqslant}
\renewcommand{\geq}{\geqslant}
\newcommand{\NN}{\mathbb{N}}

\newcommand{\TT}{\mathcal{T}}
\newcommand{\TB}{\mathcal{T}}
\newcommand{\BT}{\mathcal{T}}
\newcommand{\FF}{\mathcal{F}}

\newcommand{\RR}{\mathbb{R}}

\begin{document}

\title{Two results on expected values of imbalance indices of phylogenetic trees}

\author{Arnau Mir  \and Francesc Rossell\'o} 
\institute{Research Institute of Health Science (IUNICS) and Department of Mathematics and  Computer Science,  University of the Balearic Islands,   E-07122 Palma de  Mallorca, Spain
\email{\{arnau.mir,cesc.rossello\}@uib.es} }
\maketitle

\begin{abstract}
We compute an explicit formula for the expected value of the Colless index of a phylogenetic tree generated under the Yule model, and an explicit formula for the expected value of the Sackin index of a phylogenetic tree generated under the uniform model. 
\end{abstract}

\section{Introduction}

A \emph{phylogenetic tree} is a representation of the shared evolutionary history of a set of extant species. From the mathematical point a view, it is a leaf-labeled rooted tree, with its leaves representing  the extant species under study, its internal nodes representing common ancestors of some of them, the root representing the most recent common ancestor of all of them, and the arcs representing direct descendance through mutations. In this paper we only consider \emph{binary} phylogenetic trees, where every internal node has exactly two children.

There are several stochastic models of the evolutionary processes that produce phylogenetic trees. Two of the most popular are the Yule and the uniform models. In the \emph{Yule}, or  \emph{Equal-Rate Markov}, model \cite{Harding71,Yule}, starting with a node, at every step a leaf is chosen randomly and uniformly, and it is replaced by a \emph{cherry} (a phylogenetic tree consisting only of a root and two leaves). Finally, the labels are assigned randomly and uniformly to the leaves once the desired number of leaves is reached. Under this model of evolution, different trees with the same number of leaves may have different probabilities. More specifically, if $T$ is a binary phylogenetic tree with $n$ leaves, and for every internal node $z$ we denote by $\kappa_T(z)$ the number of its descendant leaves, then the probability of $T$ under the Yule model is   \cite{Brown,SM01}
 $$
 P_Y(T)=\frac{2^{n-1}}{n!}\prod_{v\ \mathrm{internal}}\frac{1}{\kappa_T(v)-1}.
$$
On the other hand, the main feature of the \emph{uniform}, or \emph{Proportional to Distinguishable Arrangements},  model \cite{Rosen78} is that all phylogenetic trees with the same number of leaves have the same probability. From the point of view of tree growth, this corresponds to a process where, starting with a node labeled 1, at  the $k$-th step a new pendant arc, ending in the leaf labeled $k+1$, is added either to a new root or to some edge (with all possible locations of this new pendant arc equiprobable). Notice that this is not an explicit model of evolution, only of tree growth.

The study of the probabilistic distributions of indices associated to phylogenetic trees under different stochastic models of phylogenetic tree growth has received a lot of interest in the last decades \cite{Mooers97,Shao:90}. The ultimate goal of this line of research is to be able to take as null model some stochastic model of phylogenetic tree growth and evaluate against it the indices of a sample of phylogenetic trees reconstructed from data. Two of the most popular indices used in this connection, measuring the degree of symmetry, or \emph{balance}, of a tree, are Sackin's \cite{Sackin72} and Colless' \cite{Colless82} indices, which we define later in the main body of this paper, but there are many other measures associated to phylogenetic trees that have been used in this context, like for instance other imbalance indices \cite[Chap. 33]{fel:04} or the number of cherries of trees \cite{cherries}.  

Several properties of the distributions of Sackin's  $S$ and Colless'  $C$ indices have been studied in the literature under different models \cite{BFJ:06,KiSl:93,Mul11,Rogers:93,Rogers:94,Rogers:96,SM01}. In particular, their expected values have been studied under the Yule and the uniform model. The results published so far on these expected values have been the following.  Let $S_n$ and $C_n$ be the random variables defined by choosing a  binary phylogenetic tree $T$ with $n$ leaves and computing $S(T)$ or $C(T)$, respectively. Then:
\begin{itemize}
\item Under the Yule model,
\begin{itemize}
\item  $E_Y(S_n)=2n\sum\limits_{j=2}^{n}1/j$ \cite{KiSl:93}.

\item  $E_Y(C_n)=n\log(n)+(\gamma-1-\log(2))n+o(n)$ \cite{BFJ:06}, where $\gamma$ is the Euler constant.
\end{itemize}

\item Under the uniform model,
\begin{itemize}
\item  $E_U(S_n)\sim \sqrt{\pi}n^{3/2}$ \cite{Tak:91}.

\item  $E_U(C_n)\sim \sqrt{\pi}n^{3/2}$ \cite{BFJ:06}.
\end{itemize}
\end{itemize}
And, for instance, these are the formulas used by the R package apTreeshape  \cite{apTs} to compute the expected value of these indices for a given number of leaves. Let us also mention that Rogers \cite{Rogers:94} found a recursive formula for the moment-generating functions of $C$ and $S$, which allowed him to compute $E_Y(C_n)$ and $E_U(C_n)$ for $n=1,\ldots,50$, but he did not obtain any explicit formula for them.

In this paper we obtain explicit formulas for $E_Y(C_n)$  and $E_U(S_n)$. Namely,
$$
E_Y(C_n)=n\sum_{j=2}^{\lfloor n/2\rfloor} \frac{1}{j} +\delta_{odd}(n),
$$
where $\delta_{odd}(n)=1$ if $n$ is odd, and $\delta_{odd}(n)=0$ if $n$ is even, and
$$
E_U(S_n)=\frac{n}{2n-3} {}_3 F_2
\bigg(\begin{array}{l} 2,\ 2,\ 2-n \\[-0.5ex] 1,\  4-2n\end{array};2\bigg),
$$
where ${}_3 F_2$ is a \emph{hypergeometric function}  \cite{Bayley} that can be directly computed with many software systems, like Mathematica or R. These formulas thus contribute to our knowledge of the probability distributions of these indices, and yield precise values which can be used is tests.

\section{Preliminaries and notations}
In this paper, by a  \emph{phylogenetic tree} on a set $S$ of taxa we mean a binary rooted tree  with its leaves bijectively labeled in  the set $S$.  To simplify the language, we shall always identify a leaf of a phylogenetic tree with its label.  We shall also use the term \emph{phylogenetic tree with $n$ leaves} to refer to a phylogenetic tree on the set $\{1,\ldots,n\}$.  We shall denote by $L(T)$ the set of leaves of a phylogenetic tree $T$ and by $V_{int}(T)$ its set of internal nodes.

Let  $\TT(S)$ be the set of isomorphism classes of phylogenetic trees on a set $S$ of taxa, and 
set $\TT_n=\TT(\{1,\ldots,n\})$.
It is well known \cite[Ch.3]{fel:04} that $|\TB_1|=1$ and, for every $n\geq 2$,
$$
|\TB_n|=(2n-3)!!=(2n-3)(2n-5)\cdots 3 \cdot 1.
$$

Whenever there exists a path from $u$ to $v$ in a phylogenetic tree $T$, we shall say that $v$ is a  \emph{descendant} of $u$.  The \emph{cluster} of a node $v$ in $T$ is the set $C_T(v)$ of its descendant leaves, an we shall denote by $\kappa_T(v)$ the cardinal $|C_T(v)|$, that is, the number of descendant leaves of $v$.
The \emph{depth} $\delta_T(v)$ of a node $v$  in a phylogenetic tree $T$ is the length (number of arcs) of the unique path from the root $r$ to $v$.

Given two phylogenetic trees $T_1,T_2$ on disjoint sets of taxa $S_1,S_2$, respectively, we shall denote by $T_1\widehat{\ }\, T_2$ the tree on $S_1\cup S_2$ obtained by connecting the roots of $T_1$ and $T_2$ to a (new) common parent $r$.
Every tree in $\BT_n$ is obtained as $T_k\widehat{\ }\, {}T_{n-k}$, for some subset  $S_k\subseteq \{1,\ldots,n\}$ with $k$ elements (with $1\leq k\leq n-1$), some tree $T_k$ on $S_k$ and some tree $T_{n-k}$ on $S_k^c=\{1,\ldots,n\}\setminus S_k$: actually, if we perform  in this order the choices necessary to produce a tree $T\in \TB_n$ in this way, we
obtain every tree in $\TB_n$  {twice}.

An \emph{ordered $m$-forest} on a set $S$  is an ordered sequence of $m$ phylogenetic trees $(T_1,T_2,\ldots,T_m)$, each $T_i$ on a set $S_i$ of taxa, such that these sets $S_i$ are pairwise disjoint and their union is $S$.  Let $\FF_{m,n}$ be the set of isomorphism classes of ordered $m$-forests  on $\{1,\ldots,n\}$. It is known (see, for instance, \cite[Lem. 1]{MirR10}) that for every $n\geq m\geq 1$,
$$
|\FF_{m,n}|= \frac{(2n-m-1)!m}{(n-m)!2^{n-m}}.
$$

\section{Expected value of the Colless index under the Yule model}

Let $T$ be a phylogenetic tree. For every $v\in V_{int}(T)$,  the \emph{balance value} of $v$ is $bal_T(v)=|\kappa_T(v_1)-\kappa_T(v_2)|$, where  $v_1$ and $v_2$ are its children. The \emph{Colless index}  \cite{Colless82} of a phylogenetic tree $T\in \TT_n$ is
$$
C(T)=\sum_{v\in V_{int}(T)}bal_T(v).
$$

\begin{lemma}\label{lem:hat}
If $T_k\in \TT_k$ and $T_{n-k}\in\TT_{n-k}$, then
\begin{enumerate}[(a)]
\item $C(T_k\widehat{\ }\, {}T_{n-k})=C(T_k)+C(T_{n-k})+|2n-k|$

\item $P_Y(T_k\widehat{\ }\, {}T_{n-k})=\dfrac{2}{(n-1)\binom{n}{k}} P_Y(T_{k})P_Y(T_{n-k})$
\end{enumerate}
where $P_Y$ denotes the probability of a phylogenetic tree under the Yule model.
\end{lemma}

\begin{proof}
Assertion (a) is well known, and a direct consequence of the definition of $C$. Assertion (b) is a direct consequence of the explicit probabilities of $T_k$, $T_{n-k}$ and $T_k\widehat{\ }\, {}T_{n-k}$ under the Yule model, and the fact that $V_{int}(T_k\widehat{\ }\, {}T_{n-k})=V_{int}(T_{k})\cup V_{int}(T_{n-k})\cup\{r\}$, these unions being disjoint.
\end{proof}

\begin{lemma}\label{lem:YI}
Let $I:\bigcup_{n\in \NN}\BT_n\to \RR$ be a mapping such that, for every  phylogenetic trees
$T_1,T_2$ on disjoint sets of taxa $S_1,S_2$, respectively,
$$
I(T_1\widehat{\ }\, T_2)=I(T_1)+I(T_2)+f(|S_1|,|S_2|)
$$
for some mapping $f:\NN\times \NN\to \RR$. For every $n\geq 1$, let $I_n$ be the random variable that chooses a tree $T_n\in \TB_n$ and computes $I(T_n)$, and let $E_Y(I_n)$ be its expected value under the Yule model. Then,
$$
E_Y(I_n)=\frac{1}{n-1}\Big(2\sum_{k=1}^{n-1} E_Y(I_k)  + \sum_{k=1}^{n-1}f(k,n-k)\Big)
$$
\end{lemma}

\begin{proof}
We compute $E_Y(I_n)$ using its very definition and Lemma \ref{lem:hat}.(b):
$$
\begin{array}{l}
E_Y(I_n) \displaystyle =\sum_{T_n\in \TB_n} I(T_n)\cdot p_Y(T_n)
\\
\quad \displaystyle =\sum_{k=1}^{n-1}\sum_{S_k\subset\{1,\ldots,n\}\atop |S_k|=k}
 \sum_{T_k\in \TB(S_k)}\sum_{T_{n-k}\in \TB(S_k)}I(T_k\widehat{\ }\, {}T_{n-k})\cdot p_Y(T_k\widehat{\ }\, {}T_{n-k})
 \\ 
\quad  \displaystyle  =\frac{1}{2} \sum_{k=1}^{n-1}\binom{n}{k}
 \sum_{T_k\in \TB_k}\sum_{T_{n-k} \in \TB_{n-k}}(I(T_k)+I(T_{n-k})\\
\quad  \displaystyle  \qquad\qquad\qquad\qquad\qquad +f(k,n-k))\cdot \frac{2}{(n-1)\binom{n}{k}} P_Y(T_{k})P_Y(T_{n-k})\\
\quad  \displaystyle  =\frac{1}{n-1}\sum_{k=1}^{n-1} 
 \sum_{T_k}\sum_{T_{n-k}}(I(T_k)+I(T_{n-k})+f(k,n-k))P_Y(T_{k})P_Y(T_{n-k})
 \\
\quad  \displaystyle  =\frac{1}{n-1}\sum_{k=1}^{n-1} 
\Big( \sum_{T_k}\sum_{T_{n-k} }I(T_k)P_Y(T_{k})P_Y(T_{n-k})   \\
\quad  \displaystyle\qquad\qquad\qquad\qquad + \sum_{T_k}\sum_{T_{n-k} } I(T_{n-k}) P_Y(T_{k})P_Y(T_{n-k}) \\
\quad  \displaystyle\qquad\qquad\qquad\qquad
 + \sum_{T_k}\sum_{T_{n-k}} f(k,n-k)P_Y(T_{k})P_Y(T_{n-k}) \Big)
 \\
\quad  \displaystyle  =\frac{1}{n-1}\sum_{k=1}^{n-1} 
\Big( \sum_{T_k} I(T_k)P_Y(T_{k})  +  \sum_{T_{n-k}} I(T_{n-k})  P_Y(T_{n-k})   + f(k,n-k) \Big)
 \\
\end{array}
$$
$$
\begin{array}{l}
\quad  \displaystyle  =\frac{1}{n-1}\sum_{k=1}^{n-1} 
(E_Y(I_k) + E_Y(I_{n-k})   + f(k,n-k))
 \\
\quad  \displaystyle  =\frac{1}{n-1}\Big(2\sum_{k=1}^{n-1} 
E_Y(I_k)  + \sum_{k=1}^{n-1}f(k,n-k)\Big)
\end{array}
$$
as we claimed.
\end{proof}

Mappings $I$ satisfying the hypothesis in the previous lemma are a special case of \emph{binary recursive tree shape statistics} in the sense of \cite{Matsen}.

\begin{theorem}
Let $C_n$ be the random variable that chooses a tree $T\in \TB_n$ and computes its Colless index $C(T_n)$. Its expected value under the Yule model  is
$$
E_Y(C_n)=n\sum_{j=2}^{\lfloor n/2\rfloor} \frac{1}{j} +\delta_{odd}(n),
$$
where $\delta_{odd}(n)=1$ if $n$ is odd, and $\delta_{odd}(n)=0$ if $n$ is even.
\end{theorem}

\begin{proof}
To simplify the notations, we shall denote $E_Y(C_n)$ simply by $E_n$.
By Lemmas  \ref{lem:hat}.(a) and \ref{lem:YI},
$$
E_n =\frac{1}{n-1}\Big(2\sum_{k=1}^{n-1}  E_k + \sum_{k=1}^{n-1}|n-2k|\Big).
$$
Now a simple computation shows that
$$
\sum_{k=1}^{n-1}|n-2k|=\left\{
\begin{array}{ll}
\dfrac{n(n-2)}{2}& \mbox{ if $n$ is even}\\
\dfrac{(n-1)^2}{2}& \mbox{ if $n$ is odd}
\end{array}\right.
$$
and therefore
$$
E_n= \frac{2}{n-1} \sum_{k=1}^{n-1} E_k  +
\left\{
\begin{array}{ll}
 \dfrac{n(n-2)}{2(n-1)}& \mbox{ if $n$ is even}\\
\dfrac{n-1}{2}& \mbox{ if $n$ is odd}
\end{array}\right.
$$
In order to obtain a recurrence of order one from this expression, we distinguish the case when $n$ is even from the case when $n$ is odd.
\begin{itemize}
\item When $n$ is even
$$
E_n= \frac{2}{n-1} \sum_{k=1}^{n-1} E_k  + \frac{n(n-2)}{2(n-1)},\quad
 E_{n-1}= \frac{2}{n-2} \sum_{k=1}^{n-2} E_k  +
\dfrac{n-2}{2}
$$
and then
$$
\begin{array}{rl}
E_n & \displaystyle = 
\frac{2}{n-1} E_{n-1}+\frac{2}{n-1} \sum_{k=1}^{n-2} E_k  + \frac{n(n-2)}{2(n-1)} \\
& \displaystyle = 
\frac{2}{n-1} E_{n-1}+\frac{n-2}{n-1}\cdot \frac{2}{n-2} \sum_{k=1}^{n-2} E_k   
+\frac{n-2}{n-1}\cdot \dfrac{n-2}{2}+
\frac{n-2}{n-1} \\
& \displaystyle = 
\frac{2}{n-1} E_{n-1}+\frac{n-2}{n-1}E_{n-1}+ \frac{n-2}{n-1} 
\\
& \displaystyle = 
\frac{n}{n-1} E_{n-1}+ \frac{n-2}{n-1} 
\end{array}
$$

\item When $n$ is odd
$$
E_n= \frac{2}{n-1} \sum_{k=1}^{n-1} E_k  + \frac{n-1}{2},
\quad E_{n-1}= \frac{2}{n-2} \sum_{k=1}^{n-2} E_k  +
\dfrac{(n-1)(n-3)}{2(n-2)}
$$
and then
$$
\begin{array}{rl}
E_n & \displaystyle = 
\frac{2}{n-1} E_{n-1}+\frac{2}{n-1} \sum_{k=1}^{n-2} E_k  + \frac{n-1}{2}\\
& \displaystyle = 
\frac{2}{n-1} E_{n-1}+\frac{n-2}{n-1}\cdot \frac{2}{n-2} \sum_{k=1}^{n-2} E_k   
+\frac{n-2}{n-1}\cdot \dfrac{(n-1)(n-3)}{2(n-2)}+1\\
& \displaystyle = 
\frac{2}{n-1} E_{n-1}+\frac{n-2}{n-1}E_{n-1}+1
\\
& \displaystyle = 
\frac{n}{n-1} E_{n-1}+1
\end{array}
$$
\end{itemize}
So, in summary,
$$
E_n=\frac{n}{n-1} E_{n-1}  +
\left\{
\begin{array}{ll}
\dfrac{n-2}{n-1} & \mbox{ if $n$ is even}\\
1 & \mbox{ if $n$ is odd}
\end{array}\right.
$$
In particular, if $n$ is even,
$$
\begin{array}{rl}
E_n & \displaystyle =\frac{n}{n-1} E_{n-1}  + \dfrac{n-2}{n-1}=
\frac{n}{n-1}\Big(\frac{n-1}{n-2} E_{n-2}  +
1\Big)  + \dfrac{n-2}{n-1}\\ & \displaystyle =
\frac{n}{n-2} E_{n-2}  +2
\end{array}
$$
Setting $x_n=E_n/n$, this equation becomes
$$
x_n=x_{n-2}+\frac{2}{n}
$$
whose solution (for even numbered terms) with $x_2=E_2/2=0$ is
$$
x_n=\sum_{i=2}^{n/2} \frac{1}{i}.
$$
Therefore, when $n$ is even,
$$
E_n=n\sum_{i=2}^{n/2} \frac{1}{i},
$$
and when $n$ is odd,
$$
E_n=\frac{n}{n-1} E_{n-1}+1=
\frac{n}{n-1}\cdot (n-1)\sum_{i=2}^{(n-1)/2} \frac{1}{i}+1=
1+n\sum_{i=2}^{\lfloor n/2\rfloor} \frac{1}{i}
$$
as we claimed.
\end{proof}

\section{Expected value of the Sackin index under the uniform model}

The \emph{Sackin index}  \cite{Sackin72} of a phylogenetic tree $T\in \TT_n$ is defined as the sum of the depths of its leaves:
$$
S(T)=\sum_{i=1}^n\delta_T(i).
$$
Alternatively, 
$$
S(T)=\sum_{v\in V_{int}(T)} \kappa_T(v).
$$

Let $S_n$ be the random variable that chooses a tree $T\in \TB_n$  and computes its Sackin index $S(T)$. Since, under the uniform model, all trees in $\TT_n$ have probability $1/((2n-3)!!)$,  the expected value of $S_n$  under the uniform model is
$$
\frac{\sum_{T\in \TB_n} S(T)}{(2n-3)!!}.
$$
So, we need to compute the numerator in this fraction. Now, for every $k=1,\ldots,n-1$, let
$$
c_{k,n}=|\{ T\in \TB_n\mid \delta_T(1)=k\}|.
$$

\begin{lemma}
For every $n\geq 3$, $\displaystyle \sum_{T\in \TB_n} S(T)=n\sum_{k=1}^{n-1} k\cdot c_{k,n}$
\end{lemma}

\begin{proof}
Notice that, for every $1\leq i\leq n$,
$$
|\{ T\in \TT_n\mid \delta_T(i)=k\}|=|\{ T\in \TT_n\mid \delta_T(1)=k\}|.
$$
Then
$$
\begin{array}{rl}
\displaystyle \sum_{T\in \TB_n} S(T) & \displaystyle =\sum_{T\in \TB_n}\sum_{i=1}^{n} \delta_T(i) =\sum_{i=1}^{n}\sum_{T\in \TB_n} \delta_T(i)\\
& \displaystyle  =\sum_{i=1}^{n}\sum_{k=1}^{n-1} k\cdot |\{ T\in \TT_n\mid \delta_T(i)=k\}|\\
& \displaystyle =\sum_{i=1}^{n}\sum_{k=1}^{n-1} k\cdot |\{ T\in \TT_n\mid \delta_T(1)=k\}|  =n\sum_{k=1}^{n-1} k\cdot c_{k,n}
\end{array}
$$
\end{proof}

\begin{lemma}
For every $n\geq 2$ and $k=1,\ldots,n-1$,
$$
c_{k,n}= \frac{(2n-k-3)!k}{(n-k-1)!2^{n-k-1}}.
$$
\end{lemma}

\begin{proof}
To compute $c_{k,n}$ for $k\geq 1$, notice that every tree $T\in \TB_n$  such that $\delta(1)=k$ will have the form described in Fig. \ref{fig:forests}. Therefore, it is determined by the ordered $k$-forest $T_1,T_2,\ldots,T_{k}$ on $\{2,\ldots,n\}$, and thus
$$
c_{k,n}=|\FF_{k,n-1}|= \frac{(2n-k-3)!k}{(n-k-1)!2^{n-k-1}}.
$$

\begin{figure}[htb]
\begin{center}
\begin{tikzpicture}[thick,>=stealth,scale=0.4]
\draw(-1,3) node[tre] (1) {}; \etq 1
\draw(2,4) node[tre] (x1) {}; 
\draw (x1)--(1);
\draw (x1)--(3,5);
\draw (3.3,5.3) node {.};
\draw (3.6,5.6) node {.};
\draw (3.9,5.9) node {.};
\draw(5,7) node[tre] (x2) {}; 
\draw (4.2,6.2)--(x2);
\draw(7,9) node[tre] (r) {}; 
\draw (x2)--(r);
\draw(5,3) node[tre] (tk) {}; 
\draw (tk)--(3.6,0)--(6.4,0)--(tk);
\draw(5,1) node  {\footnotesize $T_{k}$};
\draw (x1)--(tk);
\draw(8,6) node[tre] (t2) {}; 
\draw (t2)--(6.6,3)--(9.4,3)--(t2);
\draw(8,4) node  {\footnotesize $T_{2}$};
\draw (x2)--(t2);
\draw(10,8) node[tre] (t1) {}; 
\draw (t1)--(8.6,5)--(11.4,5)--(t1);
\draw(10,6) node  {\footnotesize $T_{1}$};
\draw (r)--(t1);
\end{tikzpicture}
\end{center}
\caption{\label{fig:forests} 
The structure of a tree $T$ with $\delta_T(1)=k$.}
\end{figure}
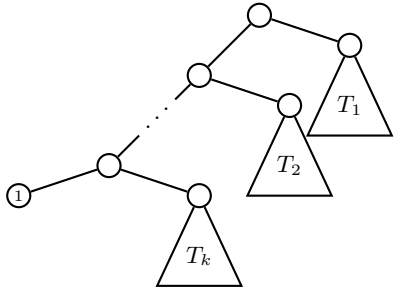
\end{proof}

Now, recall that the (\emph{generalized}) \emph{hypergeometric function} ${}_pF_q$ is defined  \cite{Bayley} as 
$$
{}_pF_q\bigg(\begin{array}{rrr} a_1,&\ldots, &a_p \\[-0.5ex] b_1,& \ldots,  &b_q\end{array};z\bigg)=\sum_{k\geq 0} \frac{(a_1)_k\cdots (a_p)_k}{(b_1)_k\cdots (b_q)_k}\cdot \frac{z^k}{k!},
$$ 
where $(a)_k := a\cdot (a+1)\cdots (a+k-1)$. Many popular software systems, like Mathematica or R, have implementations of these functions.

\begin{theorem}
The expected value of the random variable $S_n$ under the uniform model  is
$$
E_U(S_n)=\frac{n}{2n-3} {}_3 F_2
\bigg(\begin{array}{l} 2,\ 2,\ 2-n \\[-0.5ex] 1,\  4-2n\end{array};2\bigg)
$$
\end{theorem}

\begin{proof}
As we have already mentioned,
$$
E_U(S_n)=\frac{\sum_{T\in \TB_n} S(T)}{(2n-3)!!}=\frac{n}{(2n-3)!!} \sum_{k=1}^{n-1} k\cdot c_{k,n}
=\frac{n}{(2n-3)!!} \sum_{k=1}^{n-1}  \frac{(2n-k-3)!k^2}{(n-k-1)!2^{n-k-1}}
$$
Now
$$
\begin{array}{rl}
\displaystyle \frac{nk^2(2n-k-3)!}{(2n-3)!!(n-k-1)!2^{n-k-1}} & \displaystyle =
  \frac{nk^2(2n-k-3)!2^{n-2}(n-2)!}{(2n-3)!(n-k-1)!2^{n-k-1}}\\[2ex] & \displaystyle=
  \frac{nk^2(2n-k-3)!2^{n-2}(n-2)!k!}{(2n-3)!(n-k-1)!2^{n-k-1}k!} \\[2ex] & \displaystyle=
 \frac{nk^22^{k-1}\binom{n-1}{k}}{(n-1)\binom{2n-3}{k}}
\end{array}
$$
and thus
$$
\begin{array}{rl}
E_U(S_n) & = \displaystyle
\frac{n}{n-1}\sum_{k=1}^{n-1}  k^22^{k-1}\cdot \frac{\binom{n-1}{k}}{\binom{2n-3}{k}}\\
 & = \displaystyle\frac{n}{n-1} \sum_{k=1}^{n-1} \frac{k^22^{k-1}(n-1)(n-2)(n-3) \cdots (n-k)}{(2n-3)(2n-4)(2n-5)\cdots (2n-k-2)}\\
 & = \displaystyle\frac{n}{2n-3} \sum_{k=1}^{n-1} \frac{k^22^{k-1}(n-2)(n-3) \cdots (n-k)}{(2n-4)(2n-5)\cdots (2n-k-2)}\\
 & = \displaystyle\frac{n}{2n-3} \sum_{k=1}^{n-1} \frac{k^22^{k-1}(2-n)(2-n+1)\cdots (-n+k)}{(4-2n)(4-2n+1)\cdots (2-2n+k)}\\
 & = \displaystyle\frac{n}{2n-3} \sum_{k=0}^{n-2} \frac{(k+1)^22^k(2-n)(2-n+1)\cdots (1-n+k)}{(4-2n)(4-2n+1)\cdots (3-2n+k)}\\
 & = \displaystyle\frac{n}{2n-3} \sum_{k\geq 0} \frac{((k+1)!)^2(2-n)(2-n+1)\cdots(1-n+k)\cdot 2^k}{(k!)^2(4-2n)(4-2n+1)\cdots (3-2n+k)}\\
 & = \displaystyle\frac{n}{2n-3} \sum_{k\geq 0} \frac{(2)_k(2)_k(2-n)_k}{(1)_k((4-2n)_k}\cdot \frac{2^k}{k!}  =\dfrac{n}{2n-3} {}_3 F_2
\bigg(\begin{array}{l} 2,\ 2,\ 2-n \\[ -0.5ex] 1,\  4-2n\end{array};2\bigg)
\end{array}
$$
as we claimed.
\end{proof}

\section{Conclusion}
In this paper we have obtained explicit formulas for the expected value of the Sackin index under the uniform model and the Colless index under the Yule model. These results add up to the already known expected value of the Sackin index under the Yule model \cite{KiSl:93}.
For any $n$, these expected values are easily computed directly using for instance the software system R, and can be used instead of their estimations in packages like apTreeshape \cite{apTs} or SymmeTREE \cite{symT}.

It remains open the problem of finding an explicit formula for the expected value of the Colless index under the uniform model.
\medskip

\noindent\textbf{Acknowledgements.}  The research reported in this paper has been partially supported by the Spanish government and the UE FEDER program, through projects MTM2009-07165 and TIN2008-04487-E/TIN. We thank G. Cardona and M. Llabr\'es for several comments on this work.

\end{document}